\documentclass[envcountsect,envcountsame]{llncs}

\usepackage[utf8]{inputenc}
\usepackage[T1]{fontenc}
\usepackage{graphicx}
\usepackage{amssymb,amsmath}
\usepackage{array}
\usepackage{mathtools}

\usepackage{bbding}
\usepackage{xspace}

\usepackage{enumitem}
\newenvironment{titemize}
	{\begin{itemize}[topsep=1ex]}
	{\end{itemize}}

\spnewtheorem{properties}[theorem]{Properties}{\bfseries}{\rmfamily}
\spnewtheorem{myquestion}[theorem]{Question}{\bfseries}{\itshape}
\spnewtheorem{reduction}[theorem]{Reduction}{\bfseries}{\rmfamily}
\spnewtheorem{myclaim}[theorem]{Claim}{\bfseries}{\rmfamily}
\newcommand\myparagraph[1]{\smallskip\noindent{\bf #1}}

\usepackage[pdftex,unicode,hidelinks]{hyperref}
\hypersetup{breaklinks=true}
\hypersetup{pdftitle={General Caching Is Hard: Even with Small Pages}}
\hypersetup{pdfauthor={Lukáš Folwarczný and Jiří Sgall}}
\hypersetup{pdfkeywords={general caching, small pages, NP-hardness}}

\def\NP{\textsc{NP}}
\def\APX{\textsc{APX}}
\def\size{\textsc{size}}
\def\cost{\textsc{cost}}
\def\IS{\textsc{IndependentSet}\xspace}
\def\threecaching#1{\textsc{3Caching({#1})}}

\renewcommand\epsilon{\varepsilon}

\def\abs#1{{\left|#1\right|}}
\def\BB{\mathcal{B}}
\def\BBb{\overline{\mathcal{B}}}
\def\PP{\mathcal{P}}
\def\NN{\mathcal{N}}
\def\II{\mathcal{I}}
\def\TT{\mathcal{T}}
\def\IIt{\widetilde{\mathcal{I}}}
\def\aa{\bar{a}}
\def\bb{\bar{b}}

\newcommand{\decproblem}[3]{\begin{tabular}{lp{9.5cm}}
\textit{Problem}: & {\sc #1}\\
\textit{Instance}: & #2\\
\textit{Question:\hspace{0.3cm}} & #3
\end{tabular}}

\begin{document}

\title{General Caching Is Hard: Even with Small Pages%
\thanks{Partially supported by the Center of Excellence -- Institute for
Theoretical Computer Science, Prague (project P202/12/G061 of
GA~\v{C}R).}
}

\author{Luk\'a\v{s} Folwarczn\'y \and Ji\v{r}\'i Sgall}

\institute{
Computer Science Institute of Charles University,\\
Faculty of Mathematics and Physics,\\
Malostransk\'e n\'am.\ 25, CZ-11800 Praha 1,\\
Prague, Czech Republic.\\
\email{\{folwar,sgall\}@iuuk.mff.cuni.cz}
}

\maketitle

\begin{abstract}
\emph{Caching} (also known as \emph{paging}) is a~classical problem concerning
page replacement policies in two-level memory systems. \emph{General caching} is
the variant with pages of different sizes and fault costs. We give the first
\NP-hardness result for general caching with small pages: General caching is
(strongly) \NP-hard even when page sizes are limited to $\{1, 2, 3\}$. It holds
already in the \emph{fault model} (each page has unit fault cost) as well as in
the \emph{bit model} (each page has the same fault cost as size). We also give
a~very short proof of the strong \NP-hardness of general caching with page sizes
restricted to $\{1, 2, 3\}$ and arbitrary costs.

\end{abstract}

\section{Introduction}

\emph{Caching} (also known as \emph{uniform caching} or \emph{paging}) is
a~classical problem in the area of online algorithms and has been extensively
studied since 1960s. It models a~two-level memory system: There is the fast
memory of size~$C$ (the \emph{cache}) and a~slow but large main memory where all
data reside. The problem instance comprises a~sequence of requests, each
demanding a~page from the main memory. No cost is incurred if the requested page
is present in the cache (a \emph{cache hit}). If the requested page is not
present in the cache (a \emph{cache fault}), the page must be loaded at the
fault cost of one; some page must be evicted to make space for the new one when
there are already $C$~pages in the cache. The natural objective is to evict
pages in such a~way that the total fault cost is minimized. For a~reference on
classical results, see Borodin and El-Yaniv~\cite{borodin_el_yaniv}.

In 1990s, with the advent of World Wide Web, a~generalized variant called
\emph{file caching} or simply \emph{general caching} was
studied~\cite{irani_multisize,young_online_file_caching}. In this setting, each
page~$p$ has its $\size(p)$ and $\cost(p)$. It costs $\cost(p)$ to load this
page into the cache and the page occupies $\size(p)$ units of memory there.
Uniform caching is the special case satisfying $\size(p) = \cost(p) = 1$ for
every page~$p$. Other important cases of this general model are
\begin{titemize}
\item the \emph{cost model} (\emph{weighted caching}):
$\size(p) = 1$ for every page~$p$;
\item the \emph{bit model}: $\cost(p) = \size(p)$ for every page~$p$;
\item the \emph{fault model}: $\cost(p) = 1$ for every page~$p$.
\end{titemize}
General caching is still a~thriving area of research as may be documented by
the fact that a~randomized online algorithm with the optimal asymptotic
competitive ratio $\mathcal{O}(\log C)$ was given only recently by
Adamaszek~et~al.~\cite{log_k_general}; this article may also serve as
a~reference on the current state of knowledge.

Caching, as described so far, requires the service to load the requested page
when a~fault occurs, which is known as caching under the \emph{forced policy}.
Allowing the service to pay the fault cost without actually loading the
requested page to the cache gives another useful and studied variant of caching,
the \emph{optional policy}.

\myparagraph{Previous work.}
In this article, we consider the problem of finding the optimal service in the
\emph{offline version} of caching when the whole request sequence is known in
advance. Uniform caching is solvable in polynomial time with a~natural algorithm
known as Belady's rule~\cite{belady}. Caching in the cost model is a~special
case of the \emph{$k$-server problem} and is also solvable in polynomial
time~\cite{marek-servers}. In late 1990s, the questions about the complexity
status of general caching were raised. The situation was summed up by Albers et
al.~\cite{albers-general-caching}: \textit{``The hardness results for caching
problems are very inconclusive. The \NP-hardness result for the Bit model uses
a~reduction from \textsc{partition}, which has pseudopolynomial algorithms. Thus
a~similar algorithm may well exist for the Bit model. We do not know whether
computing the optimum in the Fault model is \NP-hard.''}

There was no improvement until a breakthrough in 2010 when Chrobak et
al.~\cite{marek-hardness} showed that general caching is strongly
\NP-hard, already in the case of the fault model as well as in the
case of the bit model.
General caching is usually studied under the assumption that the largest page size is very
small in comparison with the total cache size, as is for example the case of the
aforementioned article by Albers et al.~\cite{albers-general-caching}.
Instances of caching with pages larger than half of the cache size (so called
obstacles) are required in the proof given by Chrobak et al.
Therefore, this hardness result is in fact still quite inconclusive.

\myparagraph{Contribution.}
We give a~novel proof of strong \NP-hardness for general caching which gives the
first hardness result restricted to small pages:
\begin{theorem}
\label{thm:hardness}
General caching is strongly \NP-hard even in the case when the page
sizes are limited to \{1,\,2,\,3\}, for both the fault model and the
bit model, and under each of the forced and optional policies.
\end{theorem}
The proof of the result for general costs (and sizes $\{1, 2, 3\}$) is
rather simple, in particular significantly simpler than the one given by
Chrobak et al.~\cite{marek-hardness}. The reductions for the result in
the fault and bit models are significantly more involved and require a
non-trivial potential-function-like argument.

\myparagraph{Relation to interval packing.}
As observed by Chrobak~et~al.~\cite{marek-hardness}, there is a~tight connection
between caching under the optional policy and the problem called \emph{interval
packing}. A~set of weighted intervals together with a~limit~$W$ is given in the
problem and the target is either to (a) choose a~subset of intervals of the
maximum cardinality such that at each point the total weight of intersected
intervals is limited by $W$ (corresponds to caching in the fault model), or (b)
choose a~subset of intervals of the maximum total weight such that at each point
the total weight of intersected intervals is limited by $W$ (corresponds to the
bit model). The problem is a~special case of the \emph{unsplittable flow problem
on line graphs}~\cite{bansal-unsplittable-flows-approx}. Our proofs can be
considered as proofs of strong \NP-hardness of either version of interval
packing. See Chrobak~et~al.~\cite{marek-hardness} for details on the connection
between interval packing and caching.

\myparagraph{Open problems.}
We prove that general caching with page sizes $\{1, 2, 3\}$ is strongly \NP-hard
while general caching with unit page sizes is easily polynomially solvable.
Closing the gap will definitely contribute to a~better understanding of the
problem:
\begin{myquestion}
Is general caching also (strongly) \NP-hard when page sizes are limited to
\{1,\,2\}? Can caching with page sizes \{1,\,2\} be solved in polynomial time,
at least in the bit or fault model?
\end{myquestion}
In a~broader perspective, the complexity status of general caching is still far
from being well understood as the best known result on approximation is
a~4-approximation due to Bar-Noy~et~al.~\cite{four-approximation} and there is
no result on the hardness of approximation. Therefore, better understanding of
the complexity of general caching remains an important challenge.
\begin{myquestion}
Is there an algorithm for general caching with an approximation ratio better
than~4, or even a~PTAS? Is general caching \APX-hard?
\end{myquestion}

\myparagraph{Outline.} The main part of our work -- a~polynomial-time reduction
from independent set to caching in the fault model under the optional policy
with page sizes restricted to \{1, 2, 3\} -- is explained in
Section~\ref{sec:reduction} and its validity is proven in
Section~\ref{sec:proof}. In Section~\ref{sec:bit_model}, we show how to modify
the reduction so that it works for the bit model as well. In
Section~\ref{sec:forced}, we show how to obtain the hardness results also for
the forced policy. Finally, we give a self-contained presentation of the simple
proof of strong \NP-hardness for general costs (in fact, only two different and
polynomial costs are needed) in Appendix~\ref{appendix:simple}.

\section{Reduction}
\label{sec:reduction}

The decision problem \IS is well-known to be \NP-complete. By
\threecaching{forced} and \threecaching{optional} we denote the decision
versions of caching under each policy with page sizes restricted to
$\{1, 2, 3\}$.

\vskip 0.2cm
\decproblem{\IS}
{A~graph $G$ and a~number~$K$.}
{Is there an independent set of cardinality~$K$ in $G$?}
\vskip 0.2cm
\decproblem{\threecaching{\textit{policy}}}
{A~universe of pages, a~sequence of page requests, numbers $C$ and
$L$. For each page~$p$ it holds $\size(p) \in \{1,2,3\}$.}
{Is there a~service under the policy \textit{policy} of the request sequence using
the cache of size~$C$ with a~total fault cost of at most~$L$?}
\vskip 0.2cm

We define \threecaching{fault,\textit{policy}} to be the problem
\threecaching{\textit{policy}} with the additional requirement that page costs
adhere to the fault model. The problem \threecaching{bit,\textit{policy}} is
defined analogously.

In this section, we describe a~polynomial-time reduction from \IS to
\threecaching{fault,optional}. Informally, a~set of pages of size two and three
is associated with each edge and a~page of size one is associated with each
vertex. Each vertex-page is requested only twice while there are many requests
on pages associated with edges. The request sequence is designed in such a~way
that the number of vertex-pages that are cached between the two requests in the
optimal service is equal to the the size of the maximum independent set.

We now show the request sequence of caching corresponding to the graph given in
\IS with a~parameter~$H$. In the next section, we prove that it is possible to
set a~proper value of $H$ and a~proper fault cost limit~$L$ such that the
reduction becomes a~valid polynomial-time reduction.

\begin{reduction}
\label{reduction:fault}
Let $G = (V,E)$ be the instance of \IS. The graph~$G$ has $n$~vertices and
$m$~edges and there is an arbitrary fixed order of edges $e_1, \ldots, e_m$.
Let $H$ be a~parameter bounded by a~polynomial function of $n$.

A~corresponding instance~$\II_G$ of \threecaching{fault,optional} is an instance
with the cache size $C = \mathit{2}mH + 1$ and the total of $\mathit{6}mH +
n$~pages. The structure of the pages and the requests sequence is described
below.
\end{reduction}

\myparagraph{Pages.}
For each vertex~$v$, we have a~vertex-page $p_v$ of size one. For each edge~$e$,
there are $6H$ edge-pages associated with it that are divided into $H$~groups.
The $i$th group consists of six pages $\aa^e_i, \alpha^e_i, a^e_i, b^e_i,
\beta^e_i, \bb^e_i$ where pages $\alpha^e_i$ and $\beta^e_i$ have size three and
the remaining four pages have size two.

For a~fixed edge~$e$, let $\aa^e$-pages be all pages $\aa^e_i$ for $i = 1,
\ldots, H$. Let also $\aa$-pages be all $\aa^e$-pages for $e = e_1, \ldots,
e_m$. The remaining collections of pages ($\alpha^e$-pages, $\alpha$-pages,
\ldots) are defined in a~similar fashion.

\myparagraph{Request sequence.}
The request sequence of~$\II_G$ is organized in phases and blocks. There is one
phase for each vertex~$v \in V$, we call such a~phase the $v$-phase. There are
exactly two requests on each vertex-page~$p_v$, one just before the beginning of
the $v$-phase and one just after the end of the $v$-phase; these requests do not
belong to any phase. The order of phases is arbitrary. In each $v$-phase, there
are $2H$~adjacent blocks associated with every edge~$e$ incident with $v$; the
blocks for different incident edges are ordered arbitrarily. In addition, there
is one initial block~$I$ before all phases and one final block~$F$ after all
phases. Altogether, there are $d = 4mH+2$ blocks.

Let $e = \{u, v\}$ be an edge, let us assume that the $u$-phase precedes the
$v$-phase. The blocks associated with~$e$ in the $u$-phase are denoted by
$B^e_{1,1}$, $B^e_{1,2}, \dots, B^e_{i,1}$, $B^e_{i,2}, \dots, B^e_{H,1}$,
$B^e_{H,2}$, in this order, and the blocks in the $v$-phase are denoted by
$B^e_{1,3}$, $B^e_{1,4}, \dots, B^e_{i,3}$, $B^e_{i,4}, \dots, B^e_{H,3}$,
$B^e_{H,4}$, in this order. An example is given in
Fig.~\ref{fig:phases-and-blocks}.

\vspace{1ex}

\begin{figure}[bth]
\centering
\includegraphics[scale=1]{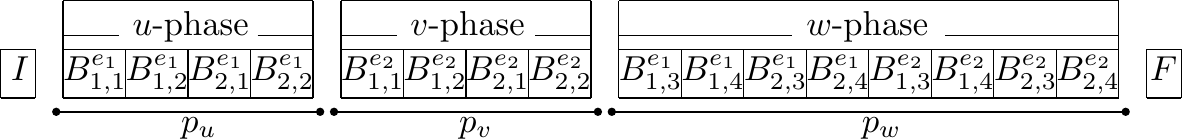}
\caption{An example of phases, blocks and requests on vertex-pages for a~graph
with three vertices $u$, $v$, $w$ and two edges $e_1 = \{u, w\}$, $e_2 = \{v,
w\}$ when $H = 2$}
\label{fig:phases-and-blocks}
\end{figure}

\vspace{2ex}
Even though each block is associated with some fixed edge, it contains one or
more requests to the associated pages for every edge $e$. In each block, we
process the edges in the order $e_1,\ldots,e_m$ that was fixed above. Pages
associated with the edge~$e$ are requested in two rounds. In each round, we
process groups $1, \ldots, H$ in this order. When processing the $i$th group of
the edge $e$, we request one or more pages of this group, depending on the block
we are in. Table~\ref{tab:requests} determines which pages are requested.

\newcolumntype{M}[1]{>{\centering\arraybackslash}m{#1}}
\newcolumntype{N}{@{}m{0pt}@{}}
\def\rowend{\\[6pt]}
\def\roundsep{\parbox[c]{.5em}{\tiny\textbullet}}

\begin{table}
\caption{Requests associated with an edge~$e$}
\vspace{-3ex}
\label{tab:requests}
\begin{center}
\begin{tabular}{|M{4cm}|M{2.5cm}M{0.2cm}M{2.5cm}|N}
\hline
Block & First round & \roundsep & Second round & \\[4pt]
\hline
before $B^e_{i,1}$ & $\aa^e_i$ & \roundsep & \rowend
\hline
$B^e_{i,1}$ & $\aa^e_i$, $\alpha^e_i$ & \roundsep & $b^e_i$ & \rowend
\hline
$B^e_{i,2}$ & $\alpha^e_i$, $a^e_i$ & \roundsep & $b^e_i$ & \rowend
\hline
between $B^e_{i,2}$ and $B^e_{i,3}$ & $a^e_i$ & \roundsep & $b^e_i$ & \rowend
\hline
$B^e_{i,3}$ & $a^e_i$ & \roundsep & $b^e_i$, $\beta^e_i$ & \rowend
\hline
$B^e_{i,4}$ & $a^e_i$ & \roundsep & $\beta^e_i$, $\bb^e_i$ & \rowend
\hline
after $B^e_{i,4}$ & & \roundsep & $\bb^e_i$ & \rowend
\hline
\end{tabular}
\end{center}
\vspace*{-1ex}
\end{table}

Reduction~\ref{reduction:fault} is now complete. An example of requests on
edge-pages associated with one edge~$e$ is depicted in
Fig.~\ref{fig:edge-pages-h3}. Notice that the order of the pages associated
with~$e$ is the same in all blocks; more precisely, in each block the requests
on the pages associated with~$e$ form a subsequence of
\begin{equation}
\label{eq:ordering}
\aa^e_1\, \alpha^e_1\, a^e_1\, \ldots\, \aa^e_i\, \alpha^e_i\, a^e_i\, \ldots\, \aa^e_H\,
\alpha^e_H\, a^e_H\; b^e_1\, \beta^e_1\, \bb^e_1\, \ldots\, b^e_i\, \beta^e_i\, \bb^e_i\, \ldots\,
b^e_H\, \beta^e_H\, \bb^e_H.
\end{equation}

\myparagraph{Preliminaries for the proof.}
Instead of minimizing the service cost, we maximize the savings compared to the
service which does not use the cache at all. This is clearly equivalent when
considering the decision versions of the problems.

Without loss of generality, we may assume that any page is brought into the
cache only immediately before some request to that page and removed from the
cache only after some (possibly different) request to that page; furthermore, at
the beginning and at the end the cache is empty. I.e., a~page may be in the
cache only between two consecutive requests to this page, and either remains in
the cache for the whole interval or not at all.

Each page of size~three is requested only twice in two consecutive blocks, and
these blocks are distinct for all pages of size three. Thus, a~service of
edge-pages is valid if and only if at each time, at most $mH$~edge-pages are in
the cache. It is convenient to think of the cache as of $mH$~\emph{slots} for
edge-pages.

As each vertex-page is requested twice, the savings on the $n$ vertex-pages are
at most~$n$. Furthermore, a vertex-page can be cached if and only if during the
phase it never happens that at the same time all slots for edge-pages are full
and a~page of size three is cached.

Let $S_B$ denote the set of all edge-pages cached at the beginning of the
block~$B$ and let $S_B^e$ be the set of pages in $S_B$ associated with the
edge~$e$. We use $s_B = |S_B|$ and $s_B^e = |S_B^e|$ for the sizes of the sets.
Each edge-page is requested only in a~contiguous segment of blocks, once in each
block. It follows that the total savings on edge-pages are equal to $\sum_B s_B$
where the sum is over all blocks. In particular, the maximal possible savings on
the edge-pages are $(d-1)mH$, using the fact that $S_I$ is empty. We shall show
that the maximum savings are $(d-1)mH + K$ where $K$ is the size of the maximum
independent set in~$G$.

\begin{figure}
\centering
\includegraphics[scale=1]{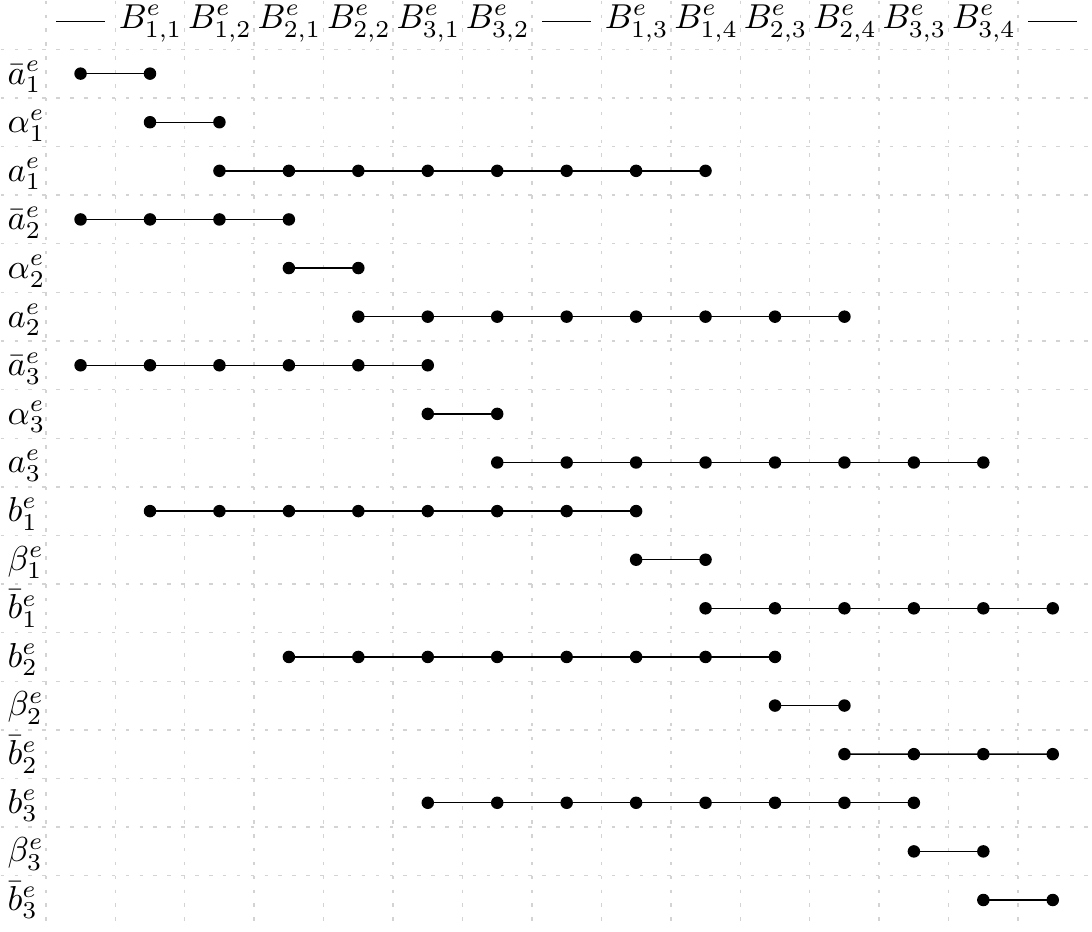}
\caption{Requests on all pages associated with the edge~$e$ when $H = 3$. Each
column represents some block(s). The labelled columns represent the blocks in
the heading, the first column represents every block before $B^e_{1,1}$, the
middle column represents every block between $B^e_{3,2}$ and $B^e_{1,3}$, and
the last column represents every block after $B^e_{3,4}$. The requests in one
column are ordered from top to bottom.}
\label{fig:edge-pages-h3}
\end{figure}

\myparagraph{Almost-fault model.}
To understand the reduction, we consider what happens if we relax the
requirements of the fault model and set the cost of each vertex-page to
$1/(n+1)$ instead of 1 as required by the fault model.

In this scenario, the total savings on vertex-pages are $n/(n+1) < 1$ which is
less than savings incurred by one edge-page. Therefore, edge-pages must be
served optimally in the optimal service of the whole request sequence.

In this case, the reduction works already for $H = 1$. This leads to a~quite
short proof of the strong \NP-hardness for general caching and we give this
proof in Appendix~\ref{appendix:simple}. Here, we show just the main ideas that
are important also for the design of our caching instance in the fault and bit
models.

We first prove that for each edge~$e$ and each block~$B \neq I$ we have $s^e_B =
1$ (see Appendix~\ref{appendix:simple}). Using this we show below that for each
edge~$e$, at least one of the pages $\alpha^e_1$ and $\beta^e_1$ is cached
between its two requests. This implies that the set of all vertices~$v$ such
that $p_v$ is cached between its two requests is independent.

For a~contradiction, let us assume that for some edge~$e$, neither of the pages
$\alpha^e_1$ and $\beta^e_1$ is cached between its two requests. Because pages
$\alpha^e_1$ and $\beta^e_1$ are forbidden, there is $b^e_1$ in $S_{B^e_{1,2}}$
and $a^e_1$ in $S_{B^e_{1,3}}$. Somewhere between these two blocks $B^e_{1,2}$
and $B^e_{1,3}$, we must switch from caching $b^e_1$ to caching $a^e_1$.
However, this is impossible, because the order of requests implies that we would
have to cache both $b^e_1$ and $a^e_1$ at some moment (see
Fig.~\ref{fig:edge-pages-h1}). However, there is no place in the cache for such
an operation, as $s^{e'}_B = 1$ for every $e'$ and $B \neq I$.

\begin{figure}[bth]
\begin{center}
\includegraphics[scale=1]{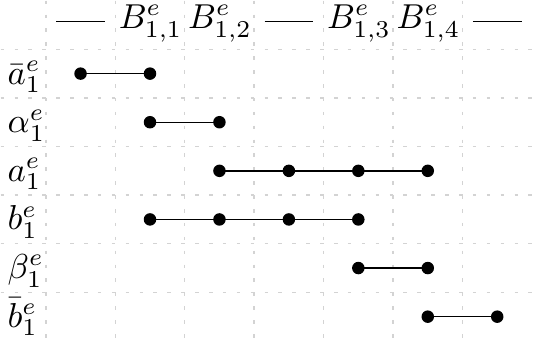}
\caption{Pages associated with one edge when $H = 1$}
\label{fig:edge-pages-h1}
\end{center}
\end{figure}

In the fault model, the corresponding claim $s_B^e = H$ does not hold. Instead,
we prove that the value of $s_B^e$ cannot change much during the service and
when we use $H$ large enough, we still get a~working reduction.

\section{Proof of Correctness}
\label{sec:proof}

In this section, we show that the reduction described in the previous section is
indeed a~reduction from \IS set to \threecaching{fault,optional}. We prove that
there is an independent set of cardinality~$K$ in~$G$ if and only if there is
a~service of the caching instance $\II_G$ with the total savings of at least
$(d-1)mH + K$. First the easy direction, which holds for any value of the
parameter~$H$.

\begin{lemma}
\label{l:easy_direction}
Let $G$ be a~graph and $\II_G$ the corresponding caching instance from
Reduction~\ref{reduction:fault}. Suppose that there is an independent set~$W$
of cardinality~$K$ in~$G$. Then there exists a~service of $\II_G$ with the total
savings of at least $(d-1)mH + K$.
\end{lemma}

\begin{proof}
For any edge~$e$, denote $e = \{u, v\}$ so that the $u$-phase precedes the
$v$-phase. If $u \in W$, we keep all $\aa^e$-pages, $b^e$-pages,
$\beta^e$-pages and $\bb^e$-pages in the cache from the first to the last
request on each page, but we do not cache $a^e$-pages and $\alpha^e$-pages at
any time. Otherwise, we cache all $\aa^e$-pages, $\alpha^e$-pages, $a^e$-pages
and $\bb^e$-pages, but do not cache $b^e$-pages and $\beta^e$-pages at any time.
Fig.~\ref{fig:easy-direction} shows these two cases for the first group of
pages. In both cases, at each time at most one page associated with each group
of each edge is in the cache and the savings on those pages are $(d-1)mH$. We
know that the pages fit in the cache because of the observations made in
Section~\ref{sec:reduction}.

For any $v\in W$, we cache $p_v$ between its two requests. To check that this
is a~valid service, observe that if $v\in W$, then during the corresponding
phase no page of size three is cached. Thus, the page~$p_v$ always fits in the
cache together with at most $mH$~pages of size two.
\qed
\end{proof}

\begin{figure}[bth]
\begin{center}
\includegraphics[scale=1]{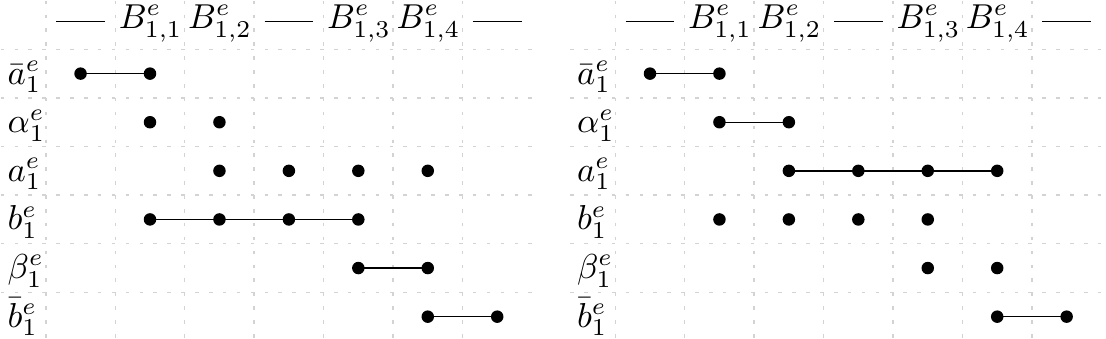}
\caption{The two ways of caching in Lemma~\ref{l:easy_direction}}
\label{fig:easy-direction}
\end{center}
\end{figure}

We prove the converse in a~sequence of lemmata. In Section~\ref{sec:bit_model},
we will show how to reuse the proof for the bit model. To be able to do that, we
list explicitly all the assumptions about the caching instance that are used in
the following proofs.

\begin{properties}
\label{the_properties}
Let~$\TT_G$ be an instance of general caching corresponding to a~graph $G =
(V,E)$ with $n$~vertices, $m$~edges $e_1, \ldots, e_m$, the same cache size and
the same universe of pages as in Reduction~\ref{reduction:fault}. The request
sequence is again split into phases, one phase for each vertex. Each phase is
again partitioned into blocks, there is one initial block~$I$ before all phases
and one final block~$F$ after all phases. There is the total of $d$~blocks.

The instance~$\TT_G$ is required to fulfill the following list of properties:
\vspace{-1ex}
\begin{enumerate}[label=(\alph*)]
\item \label{prprt:vertex_savings} Each vertex page~$p_v$ is requested exactly
twice, right before the $v$-phase and right after the $v$-phase.
\item \label{prprt:count_savings} The total savings incurred on edge-pages are
equal to $\sum s_B$ (summing over all blocks).
\item \label{prprt:blocks_I_F} For each edge~$e$, there are exactly~$H$ pages
associated with~$e$ requested in $I$, all the $\aa^e$-pages, and exactly~$H$
pages associated with~$e$ requested in $F$, all the $\bb^e$-pages.
\item \label{prprt:all_edge_ordering} In each block, pages associated with $e_1$
are requested first, then pages associated with $e_2$ are requested and so on up
to $e_m$.
\item \label{prprt:good_ordering} For each block~$B$ and each edge~$e$, all
requests on $a^e$-pages and $\bb^e$-pages in~$B$ precede all requests on
$\aa^e$-pages and $b^e$-pages in~$B$.
\item \label{prprt:alpha_beta} Let $e = \{u, v\}$ be an edge and $p$ an
$\alpha^e$-page or $\beta^e$-page. Let $B$ be the first block and $\overline{B}$
the last block where $p$ is requested. Then $B$ and $\overline B$ are either
both in the $u$-phase or both in the $v$-phase. Furthermore, no
other page of size three associated with $e$ is requested in $B$,
$\overline{B}$, or any block between them.
\end{enumerate}
\end{properties}

\begin{lemma}
\label{l:obey_properties}
The instance from Reduction~\ref{reduction:fault} satisfies
Properties~\ref{the_properties}.
\end{lemma}

\begin{proof}
All properties \ref{prprt:vertex_savings}, \ref{prprt:count_savings},
\ref{prprt:blocks_I_F}, \ref{prprt:all_edge_ordering},
\ref{prprt:alpha_beta} follow directly from Reduction~\ref{reduction:fault} and
the subsequent observations. To prove~\ref{prprt:good_ordering}, recall that the
pages associated with an edge~$e$ requested in a~particular block always follow
the ordering~\eqref{eq:ordering}. We need to verify that when the page $a^e_i$
is requested, no page $\aa^e_j$ for $j \leq i$ is requested and that when the
page $\bb^e_i$ is requested, no $\aa^e$-page and no page $b^e_j$ for $j \leq i$
is requested. This can be seen easily when we explicitly write down the request
sequences for each kind of block, see Table~\ref{the_table}.
\qed
\end{proof}

\begin{table}[t]
\caption{Request sequences on all pages associated with an edge $e$}
\vspace{-2ex}
\label{the_table}
\begin{center}
\begin{tabular}{|M{4cm}|M{7cm}|N}
\hline
Block & First round \roundsep{} Second round & \\[4pt]
\hline
before $B^e_{1,1}$ & $\aa^e_1 \ldots \aa^e_H$ \roundsep & \rowend

$B^e_{i,1}$ & $a_1^e \ldots a_{i-1}^e \; \aa_i^e \; \alpha_i^e \; \aa_{i+1}^e
\aa_{i+2}^e \ldots \aa_H^e$ \roundsep{} $b_1^e \dots b_i^e$ & \rowend

$B^e_{i,2}$ & $a_1^e \ldots a_{i-1}^e \; \alpha_i^e \; a_i^e \; \aa_{i+1}^e
\aa_{i+2}^e \ldots \aa_H^e$ \roundsep{} $b_1^e \dots b_i^e$ & \rowend

between $B^e_{H,2}$ and $B^e_{1,3}$ & $a_1^e \ldots a_H^e$ \roundsep{} $b_1^e \ldots
b_H^e$ & \rowend

$B^e_{i,3}$ & $a_i^e \ldots a_H^e$ \roundsep{} $\bb_1^e \ldots \bb_{i-1}^e \; b_i^e \;
\beta_i^e \; b_{i+1}^e b_{i+2}^e \ldots b_H^e$ & \rowend

$B^e_{i,4}$ & $a_i^e \ldots a_H^e$ \roundsep{} $\bb_1^e \ldots \bb_{i-1}^e \;\beta_i^e \;
\bb_i^e \; b_{i+1}^e b_{i+2}^e \ldots b_H^e$ & \rowend

after $B^e_{H,4}$ & \roundsep{} $\bb^e_1 \ldots \bb^e_H$ & \rowend
\hline
\end{tabular}
\end{center}
\vspace*{-5ex}
\end{table}

For the following claims, let $\TT_G$ be an instance fulfilling
Properties~\ref{the_properties}. We fix a~service of~$\TT_G$ with the total
savings of at least $(d-1)mH$.

Let $\BB$ be the set of all blocks and $\BBb$ the set of all blocks except for
the initial and final one. For a~block~$B$, we denote the block immediately
following it by~$B'$.

We define two useful values characterizing the service for the block $B$:
$\delta_B = mH - s_B$ (the number of free slots for edge-pages at the start of
the service of the block) and $\gamma_B^e = \abs{s_{B'}^e - s_B^e}$ (the change
of the number of slots occupied by pages associated with~$e$ after requests from
this block are served).

The first easy lemma says that only a~small number of blocks can start with some
free slots in the cache.

\begin{lemma}
\label{l:deltas}
When summing over all blocks except for the initial one
$$\sum_{\mathclap{B \in \BB\setminus\{I\}}} \delta_B \leq n.$$
\end{lemma}

\begin{proof}
Using the property~\ref{prprt:count_savings} and
$s_I = 0$, the savings on edge-pages are
$$\sum_{\mathclap{B \in \BB \setminus \{I\}}} s_B = (d-1)mH
- \sum_{\mathclap{B \in\BB \setminus \{I\}}}
\delta_B.$$
The total savings are assumed to be at least $(d-1)mH$. Due to the
property~\ref{prprt:vertex_savings}, the savings on vertex-pages are at most~$n$. Claim of
the lemma follows.
\qed
\end{proof}

The second lemma states that the number of slots occupied by pages associated
with a~given edge does not change much during the whole service.

\begin{lemma}
\label{l:gammas}
For each edge $e \in E$,
$$\sum_{B \in \BBb} \gamma^e_B \leq 6n.$$
\end{lemma}

\begin{proof}
Let us use the notation $S_B^{\leq k} = S_B^{e_1} \cup \cdots \cup S_B^{e_k}$
and $s_B^{\leq k} = \Bigl|S_B^{\leq k}\Bigr|$. First, we shall prove for
each~$k \leq m$
\begin{equation}
\label{e:prfx_bnd}
\sum_{B \in \BBb} \Bigl|s_{B'}^{\leq k} - s_B^{\leq k}\Bigr| \leq 3n.
\end{equation}

Let~$\PP$ denote the set of all blocks~$B$ from $\BBb$ satisfying $s_{B'}^{\leq
k} - s_B^{\leq k} \geq 0$ and let $\NN$ denote the set of all the remaining
blocks from~$\BBb$.

As a~consequence of the property~\ref{prprt:blocks_I_F}, we get $s_{I'}^{\leq
k} \in [kH - \delta_{I'}, kH]$ and $s_{F}^{\leq k} \in [kH - \delta_F, kH]$.
So we obtain the inequality
\begin{equation}
\label{e:init_final}
s_F^{\leq k} - s_{I'}^{\leq k} \geq -\delta_F.
\end{equation}

We claim $s^{\leq k}_{B'} - s^{\leq k}_B \leq \delta_B$ for each $B \in \BBb$.
We assume for a~contradiction $s^{\leq k}_{B'} - s^{\leq k}_B > \delta_B$ for
some block~$B$. We use the property~\ref{prprt:all_edge_ordering}. Then after
processing the edge~$e_k$ in $B$, the number of edge-pages in the cache is $(s_B
- s_B^{\leq k}) + s_{B'}^{\leq k} > s_B + \delta_B = mH$. But more than
$mH$~edge-pages in the cache means a~contradiction.

The summation over all blocks from~$\PP$ gives us the first bound
\begin{equation}
\label{e:bound_p}
\sum_{B \in \PP} \left(s_{B'}^{\leq k} - s_B^{\leq k}\right) \leq \sum_{B \in \PP}
\delta_B \leq n.
\end{equation}

Using the fact $\PP \mathbin{\dot{\cup}} \NN = \BBb$ and \eqref{e:init_final},
we have $$\sum_{B \in \PP} \left(s_{B'}^{\leq k} - s_B^{\leq k}\right) + \sum_{B
\in \NN} \left(s_{B'}^{\leq k} - s_B^{\leq k}\right) = s_F^{\leq k} -
s_{I'}^{\leq k} \geq -\delta_F;$$ together with \eqref{e:bound_p}, we obtain the
second bound
\begin{equation}
\label{e:bound_n}
-\sum_{B \in \NN} \left(s_{B'}^{\leq k} - s_B^{\leq k}\right)
\leq \sum_{B \in \PP} \left(s_{B'}^{\leq k} - s_B^{\leq k}\right) + \delta_F \leq 2n.
\end{equation}

Combining the bounds \eqref{e:bound_p} and \eqref{e:bound_n}, we prove~\eqref{e:prfx_bnd}
$$\sum_{B \in \BBb} \abs{s_{B'}^{\leq k} - s_B^{\leq k}} = \sum_{B \in \PP}
\left(s_{B'}^{\leq k} - s_B^{\leq k}\right) - \sum_{B \in \NN} \left(s_{B'}^{\leq k} -
s_B^{\leq k}\right) \leq n + 2n = 3n.$$

For the edge~$e_1$, the claim of this lemma is weaker than~\eqref{e:prfx_bnd}
because $\gamma^{e_1}_B = \abs{s_{B'}^{e_1} - s_B^{e_1}}$. Proving our lemma for
$e_k$ when $k > 1$ is just a~matter of using \eqref{e:prfx_bnd} for $k-1$ and
$k$ together with the formula $\abs{x-y} \leq \abs{x} + \abs{y}$:
$$
~\qquad~
\sum_{B \in \BBb} \gamma^{e_k}_B \leq \sum_{B \in \BBb} \abs{s_{B'}^{\leq k} - s_B^{\leq
k}} + \sum_{B \in \BBb} \abs{s_{B'}^{\leq k-1} - s_B^{\leq k-1}} \leq 3n + 3n = 6n
~\qquad~
\qed
$$
\end{proof}

For the rest of the proof, we set $H = 6mn + 3n + 1$. This enables us to show
that the fixed service must cache some of the pages of size three.

\begin{lemma}
\label{l:alpha_beta}
For each edge $e \in E$, there is a~block~$B$ such that some $\alpha^e$-page or
$\beta^e$-page is in~$S_B$ and $\delta_B = 0$.
\end{lemma}

\begin{proof}
Fix an edge $e = e_k$. For each block~$B$, we define
$$\epsilon_B = \mbox{number of $\alpha^e$-pages and $\beta^e$-pages in $S_B$.}$$

Observe that due to the property~\ref{prprt:alpha_beta}, $\epsilon_B$ is always one or zero. We use
a~potential function $$\Phi_B = \mbox{number of $a^e$-pages and $\bb^e$-pages in $S_B$}.$$

Because there are only $\aa$-pages in the initial block and only $\bb$-pages in
the final block (property~\ref{prprt:blocks_I_F}), we know
\begin{equation}
\label{e:phi_ends}
\Phi_{I'} = 0 \quad\mbox{and}\quad \Phi_{F} \geq H - \delta_F.
\end{equation}

Now we bound the increase of the potential function as
\begin{equation}
\label{e:phi_increase}
\Phi_{B'} - \Phi_{B} \leq \delta_B + \sum_{\ell = 1}^{k-1} \gamma_B^{e_\ell} + \epsilon_B.
\end{equation}

To justify this bound, we fix a~block~$B$ and look at the cache state after
requests on edges $e_1, \ldots, e_{k-1}$ are processed. How many free slots
there can be in the cache? There are initial $\delta_B$ free slots in the
beginning of the block~$B$, and the number of free slots can be further
increased when the number of pages in the cache associated with $e_1, \ldots,
e_{k-1}$ decreases. This increase can be naturally bounded by $\sum_{\ell =
1}^{k-1} \gamma^{e_\ell}_B$. Therefore, the number of free slots in the cache is
at most $\delta_B + \sum_{\ell = 1}^{k-1} \gamma^{e_\ell}_B$.

Because of the property~\ref{prprt:good_ordering}, the number of cached
$a^e$-pages and $\bb^e$-pages can only increase by using the free cache space or
caching new pages instead of $\alpha^e$-pages and $\beta^e$-pages. We already
bounded the number of free slots and $\epsilon_B$ is a~natural bound for the
increase gained on $\alpha^e$-pages and $\beta^e$-pages. Thus, the
bound~$\eqref{e:phi_increase}$ is correct.

Summing \eqref{e:phi_increase} over all $B \in \BBb$, we have
$$\Phi_F - \Phi_{I'}
= \sum_{B \in \BBb} \left(\Phi_{B'} - \Phi_B\right)
\leq \sum_{B \in \BBb} \biggl( \delta_B + \sum_{\ell=1}^{k-1} \gamma^{e_\ell}_B +
\epsilon_B \biggr)$$
which we combine with \eqref{e:phi_ends} into
$$H - \delta_F \leq \sum_{B\in \BBb} \biggl( \delta_B +
\sum_{\ell=1}^{k-1} \gamma^{e_\ell}_B + \epsilon_B \biggr),$$
and use Lemmata~\ref{l:deltas} and \ref{l:gammas} to bound $\sum \epsilon_B$ as
\begin{align*}
\sum_{B \in \BBb} \epsilon_B
& \geq H - \delta_F - \sum_{B \in \BBb} \biggl( \delta_B +
\sum_{\ell=1}^{k-1} \gamma^{e_\ell}_B \biggr)\\
& \geq H - n - n - (k-1)6n\\
& \geq H - 6mn - 2n = n + 1.
\end{align*}

As there is at most one page of size three requested in each block
(property~\ref{prprt:alpha_beta}), the inequality $\sum \epsilon_B \geq n+1$
implies that there are at least $n+1$ blocks where an $\alpha^e$-page or
a~$\beta^e$-page is cached. At most $n$~blocks have $\delta_B$ non-zero
(Lemma~\ref{l:deltas}); we are done.
\qed
\end{proof}

We are ready to complete the proof of the harder direction.

\begin{lemma}
\label{l:hard_direction}
Suppose that there exists a~service of~$\TT_G$ with the total savings of at
least $(d-1)mH + K$. Then the graph~$G$ has an independent set~$W$ of
cardinality~$K$.
\end{lemma}

\begin{proof}
Let~$W$ be a~set of $K$~vertices such that the corresponding page~$p_v$ is
cached between its two requests. (There are at least $K$ of them because the
maximal savings on edge-pages are $(d-1)mH$.)

Consider an arbitrary edge~$e = \{u, v\}$. Due to Lemma~\ref{l:alpha_beta},
there exists a~block~$B$ such that $\delta_B = 0$ and some $\alpha^e$-page or
$\beta^e$-page is cached in the beginning of the block. This block~$B$ is either
in the $u$-phase or in the $v$-phase, because of the statement of the
property~\ref{prprt:alpha_beta}. This means that at least one of the two
pages~$p_u$ and $p_v$ is not cached between its two requests, because the cache
is full. As a~consequence, the set~$W$ is indeed independent.
\qed
\end{proof}

The value of $H$ was set to $6mn + 3n + 1$, therefore
Reduction~\ref{reduction:fault} is indeed polynomial.
Lemmata~\ref{l:easy_direction}, \ref{l:obey_properties} and
\ref{l:hard_direction} together imply that there is an independent set of
cardinality $K$ in $G$ if and only if there is a~service of the instance $\II_G$
with the total savings of at least $(d-1)mH + K$. We showed that the
problem~\threecaching{fault,optional} is indeed strongly \NP-hard.

\section{Bit Model}
\label{sec:bit_model}

In this section, we show how to modify the proof for the fault model from the
previous sections so that it works as a~proof for the bit model as well.

\begin{reduction}
\label{reduction:bit}
Let $G$ be a~graph and $\II_G$ the~corresponding instance of the problem
\threecaching{fault,optional} from Reduction~\ref{reduction:fault}. Then the
modified instance $\IIt_G$ is an instance of \threecaching{bit,optional} with
the same cache size and the same set of pages with the same sizes.

The structure of phases and requests on vertex-pages is also preserved. The
blocks from $\II_G$ are also used, but between each pair of consecutive blocks
there are five new blocks inserted. Let $B$ and $B'$ be two consecutive blocks.
Between $B$ and $B'$ we insert five new blocks $B_{(1)}, \dots, B_{(5)}$ with
the following requests
\begin{titemize}
\item $B_{(1)}$: do not request anything;
\item $B_{(2)}$: request all pages of size two that are requested both in $B$ and $B'$;
\item $B_{(3)}$: request a~page (there is either one or none) of size three that is requested both in
$B$ and $B'$;
\item $B_{(4)}$: request all pages of size two that are requested both in $B$ and $B'$;
\item $B_{(5)}$: do not request anything.
\end{titemize}
\end{reduction}

See Fig.~\ref{fig:instance-modification} for an example. In each new block, the
order of chosen requests is the same as in $B$ (which is the same as in $B'$, as
both follow the same ordering of edges~\eqref{eq:ordering}).
The new instance has the total of $\tilde{d} = d
+ 5(d-1) = 6d - 5$ blocks. This time we prove that the maximal total savings
are $(\tilde{d}-1)mH + K$ where $K$ is the cardinality of the maximum
independent set in~$G$.

\begin{figure}[bth]
\centering
\includegraphics[scale=1]{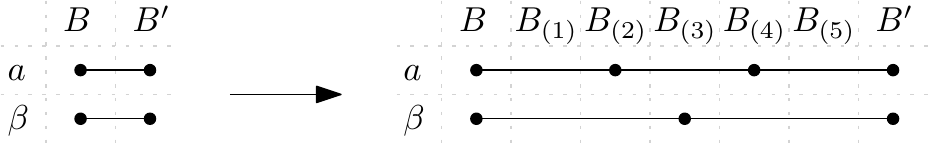}
\caption{The modification of the instance for a page~$a$ of size two and
a page~$\beta$ of size three}
\label{fig:instance-modification}
\end{figure}

\begin{lemma}
\label{l:easy_direction_bit}
Suppose that the graph~$G$ has an independent set~$W$ of cardinality~$K$. Then there
exists a~service of the modified instance~$\IIt_G$ with the total savings of at least
$(\tilde{d}-1)mH + K$.
\end{lemma}

\begin{proof}
We consider the service of the original instance~$\II_G$ described in the proof
of Lemma~\ref{l:easy_direction} and modify it so it becomes a~service of the
modified instance.

In the new service, vertex-pages are served the same way as in the original
service. The savings on vertex-pages are thus again $K$.

For each pair of consecutive blocks $B$ and $B'$, each page kept in the cache
between $B$ and $B'$ in the original service is kept in the cache in the new
service for the whole time between $B$ and $B'$ (it spans over seven blocks
now). For a~page of size two, savings of two are incurred three times. For
a~page of size three, savings of three are incurred twice. On each page in the
new service we save six instead of one. Therefore, the total savings on
edge-pages are $6(d-1)mH = (\tilde{d}-1)mH$.

The total savings are $(\tilde{d}-1)mH + K$.
\qed
\end{proof}

\begin{lemma}
\label{l:hard_direction_bit}
Suppose that there exists a~service of the modified instance~$\IIt_G$ with the
total savings of at least $(\tilde{d}-1)mH + K$. Then the graph~$G$ has an
independent set~$W$ of cardinality~$K$.
\end{lemma}

\begin{proof}
This lemma is the same as Lemma~\ref{l:hard_direction}. We just need to verify
that the modified instance fulfills Properties~\ref{the_properties}.

To prove that the property~\ref{prprt:count_savings} is preserved, we observe
that each two consecutive requests on a~page of size two are separated by
exactly one block where the page is not requested. Consequently, when there are
savings of two on a~request on the page of size two, we assign savings of one to
the block where the savings were incurred and savings of one to the previous
block. Similarly, each pair of consecutive requests on a~page of size three is
separated by exactly two blocks where the page is not requested. When there are
savings of three on a~request on the page of size three, we assign savings of
one to the block where the savings were incurred and savings of one to each of
the two previous blocks. As a~consequence, the total savings gained on
edge-pages may indeed be computed as $\sum s_B$.

The property~\ref{prprt:vertex_savings} is preserved because the requests on
vertex-pages are the same in both instances. The property~\ref{prprt:blocks_I_F}
is preserved because the initial and final blocks are the same in both
instances.

Each sequence of requests in a~block of the modified instance~$\IIt_G$ is either
the same or a~subsequence of the sequence in a~block in the original instance.
Therefore, the properties \ref{prprt:all_edge_ordering},
\ref{prprt:good_ordering} and \ref{prprt:alpha_beta} are preserved as well.
\qed
\end{proof}

Lemmata~\ref{l:easy_direction_bit} and \ref{l:hard_direction_bit} imply that we
have a~valid polynomial-time reduction and so the problem
\threecaching{bit,optional} is strongly \NP-hard.

\section{Forced Policy}
\label{sec:forced}

\begin{theorem}
\label{thm:optional_to_forced}
Both the problem \threecaching{fault,forced} and the problem
\threecaching{bit,forced} are strongly \NP-hard.
\end{theorem}

\begin{proof}
For both the fault model and the bit model, we show a~polynomial-time reduction
from caching with optional policy to the corresponding variant of caching with
the forced policy. Let us have an instance of caching with the optional policy
with the cache size~$C$ and the request sequence $\rho = r_1 \ldots r_n$; let
$M$ be the maximal size of a page in $\rho$ (in our previous reductions, $M=3$).

We create an instance of caching with the forced policy. The cache size is $C' =
C + M$. The request sequence is $\rho' = r_1 q_1 r_2 q_2 \ldots r_n q_n$ where
$q_1, \ldots, q_n$ are requests to $n$~different pages that do not appear
in~$\rho$ and have size $M$. The costs of the new pages are one in the fault
model and $M$ in the bit model.

We claim that there is a~service of the optional instance with savings~$S$
if and only if there is a~service of the forced instance with savings~$S$.

$\Rightarrow$ We serve the requests on original pages the same way as in the
optional instance. The cache is larger by $M$ which is the size of the largest
page. Thus, pages that were not loaded into the cache because of the optional
policy fit in there; we can load them and immediately evict them. New pages fit
into the cache as well and we also load them and immediately evict them. This
way we have the same savings as in the optional instance.

$\Leftarrow$ We construct a~service for the optional instance: For each~$i$,
when serving $r_i$ we consider the evictions done when serving~$r_i$ and $q_i$
of the forced instance. If a~page requested before $r_i$ is evicted, we evict it
as well. If a~page requested by $r_i$ is evicted, we do not cache it at all.
Because the page requested by $q_i$ has size $M$, the original pages occupy at
most $C$~slots in the cache when~$q_i$ is served. This way we obtain a~service
of the optional instance with the same savings.

Using the strong \NP-hardness of the problems \threecaching{fault,optional} and
\threecaching{bit,optional} proven in Sections~\ref{sec:reduction}
and~\ref{sec:proof} and the observation that the reduction preserves the maximal
size of a page, we obtain the strong \NP-hardness of \threecaching{fault,forced}
and \threecaching{bit,forced}.
\qed
\end{proof}

\bibliographystyle{splncs03}
\bibliography{caching}

\begin{thebibliography}{10}
\providecommand{\url}[1]{\texttt{#1}}
\providecommand{\urlprefix}{URL }

\bibitem{log_k_general}
Adamaszek, A., Czumaj, A., Englert, M., R{\"{a}}cke, H.: An \emph{O}(log
  \emph{k})-competitive algorithm for generalized caching. In: Proceedings of
  the 23rd Annual {ACM-SIAM} Symposium on Discrete Algorithms ({SODA}). pp.
  1681--1689. {SIAM} (2012),
  \url{http://dl.acm.org/citation.cfm?id=2095116.2095249}

\bibitem{albers-general-caching}
Albers, S., Arora, S., Khanna, S.: Page replacement for general caching
  problems. In: Proceedings of the 10th Annual {ACM-SIAM} Symposium on Discrete
  Algorithms. pp. 31--40 (1999),
  \url{http://dl.acm.org/citation.cfm?id=314500.314528}

\bibitem{bansal-unsplittable-flows-approx}
Bansal, N., Friggstad, Z., Khandekar, R., Salavatipour, M.R.: A logarithmic
  approximation for unsplittable flow on line graphs. {ACM} Transactions on
  Algorithms  10(1), ~1 (2014), \url{http://doi.acm.org/10.1145/2532645},
  a~preliminary version appeared at {SODA}~2009.

\bibitem{four-approximation}
Bar{-}Noy, A., Bar{-}Yehuda, R., Freund, A., Naor, J., Schieber, B.: A unified
  approach to approximating resource allocation and scheduling. Journal of the
  {ACM}  48(5),  1069--1090 (2001),
  \url{http://doi.acm.org/10.1145/502102.502107}, a~preliminary version
  appeared at STOC~2000.

\bibitem{belady}
Belady, L.A.: A study of replacement algorithms for a virtual-storage computer.
  IBM Systems Journal  5(2),  78--101 (Jun 1966),
  \url{http://dx.doi.org/10.1147/sj.52.0078}

\bibitem{borodin_el_yaniv}
Borodin, A., El{-}Yaniv, R.: Online computation and competitive analysis.
  Cambridge University Press (1998)

\bibitem{marek-servers}
Chrobak, M., Karloff, H.J., Payne, T.H., Vishwanathan, S.: New results on
  server problems. {SIAM} Journal on Discrete Mathematics  4(2),  172--181
  (1991), \url{http://dx.doi.org/10.1137/0404017}, a~preliminary version
  appeared at SODA~1990.

\bibitem{marek-hardness}
Chrobak, M., Woeginger, G.J., Makino, K., Xu, H.: Caching is hard -- even in
  the fault model. Algorithmica  63(4),  781--794 (2012),
  \url{http://dx.doi.org/10.1007/s00453-011-9502-9}, a~preliminary version
  appeared at ESA~2010.

\bibitem{irani_multisize}
Irani, S.: Page replacement with multi-size pages and applications to web
  caching. Algorithmica  33(3),  384--409 (2002),
  \url{http://dx.doi.org/10.1007/s00453-001-0125-4}, a~preliminary version
  appeared at STOC~1997.

\bibitem{young_online_file_caching}
Young, N.E.: On-line file caching. Algorithmica  33(3),  371--383 (2002),
  \url{http://dx.doi.org/10.1007/s00453-001-0124-5}, a~preliminary version
  appeared at SODA 1998.

\end{thebibliography}

\newpage
\appendix

\section{The Simple Proof}
\label{appendix:simple}

In this appendix, we present a~simple variant of the proof for the almost-fault
model with two distinct costs. This completes the sketch of the proof presented
at the end of Section~\ref{sec:reduction}. We present it with a~complete
description of the simplified reduction, so that it can be read independently
of the rest of the paper. This appendix can therefore serve as a~short proof of
the hardness of general caching.

\begin{theorem}
\label{thm:simple_hardness}
General caching is strongly \NP-hard, even in the case when page sizes are
limited to \{1,\,2,\,3\} and there are only two distinct fault costs.
\end{theorem}

We prove the theorem for the optional policy. It is easy to obtain the
theorem also for the forced policy the same way as in the proof of 
Theorem~\ref{thm:optional_to_forced}. 

\subsection*{The Reduction}

The reduction described here will be equivalent to
Reduction~\ref{reduction:fault} with $H = 1$ and the fault cost of each
vertex-page set to $1/(n+1)$.

Suppose we have a~graph $G=(V,E)$ with $n$~nodes and $m$~edges. We construct an
instance of general caching whose optimal solution encodes a~maximum
independent set in $G$. Fix an arbitrary numbering of edges $e_1, \ldots, e_m$.

The cache size is $C=2m+1$. For each vertex~$v$, we have a~vertex-page~$p_v$
with size one and cost $1/(n+1)$. For each edge $e$, we have six associated
edge-pages $a^e, \aa^e, \alpha^e, b^e, \bb^e, \beta^e$; all have cost one,
pages $\alpha^e,\beta^e$ have size three and the remaining pages have size two.

The request sequence is organized in phases and blocks. There is one
phase for each vertex. In each phase, there are two adjacent blocks
associated with every edge $e$ incident with $v$; the incident edges
are processed in an arbitrary order. In addition, there is one initial
block~$I$ before all phases and one final block~$F$ after all
phases. Altogether, there are $d=4m+2$ blocks. There are four blocks
associated with each edge $e$; denote them $B^e_1$, $B^e_2$,
$B^e_3$, $B^e_4$, in the order as they appear in the request
sequence.

For each $v\in V$, the associated page $p_v$ is requested exactly twice, right
before the beginning of the $v$-phase and right after the end of the $v$-phase;
these requests do not belong to any phase. An example of the structure of
phases and blocks is given in Fig.~\ref{fig:simple-phases-and-blocks}.

\begin{figure}[bth]
\centering
\includegraphics[scale=1]{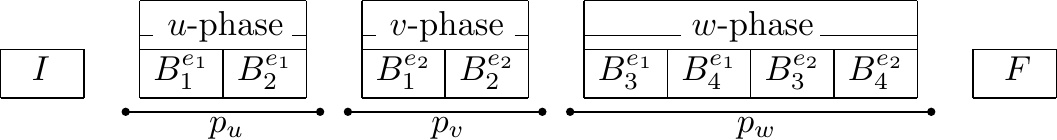}
\caption{An example of phases, blocks and requests on vertex-pages for a~graph
with three vertices $u$, $v$, $w$ and two edges $e_1 = \{u, w\}$, $e_2 = \{v,
w\}$ when $H = 2$}
\label{fig:simple-phases-and-blocks}
\end{figure}
\vspace{1ex}

Even though each block is associated with some fixed edge, it contains one or
more requests to the associated pages for every edge $e$.  In each block, we
process the edges $e_1,\ldots,e_m$ in this order. For each edge $e$, we make
one or more requests to the associated pages as follows. If the current block
is:

$\bullet$  before $B^e_1$: request $\aa^e$;

$\bullet$ $B^e_1$: request $\aa^e$, $\alpha^e$, and $b^e$;

$\bullet$ $B^e_2$: request $\alpha^e$, $a^e$, and $b^e$;

$\bullet$ after $B^e_2$ and before $B^e_3$: request $a^e$ and $b^e$;

$\bullet$ $B^e_3$: request $a^e$, $b^e$, and $\beta^e$;

$\bullet$ $B^e_4$: request $a^e$, $\beta^e$, and $\bb^e$;

$\bullet$ after $B^e_4$: request $\bb^e$.

Fig.~\ref{fig:simple-edge-pages} shows an example of the requests on edge-pages
associated with one particular edge.

\begin{figure}
\begin{center}
\includegraphics[scale=1]{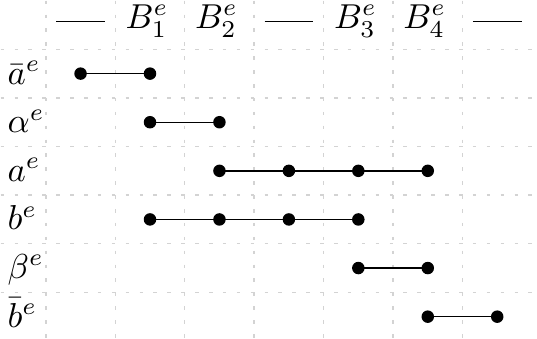}
\end{center}
\caption{Requests on all pages associated with the edge~$e$.  Each column
represents some block(s). The four labelled columns represent the blocks in the
heading, the first column represents every block before $B^e_1$, the middle
column represents every block between $B^e_3$ and $B^e_4$, and the last column
represents every block after $B^e_4$. The requests in one column are ordered
from top to bottom.}
\label{fig:simple-edge-pages}
\end{figure}

\subsection*{Proof of Correctness}

Instead of minimizing the service cost, we maximize the savings compared to the
service which does not use the cache at all. This is clearly equivalent when
considering the decision version of the problem.

Without loss of generality, we assume that any page is brought into the cache
only immediately before a request to that page and removed from the cache only
immediately after a request to that page; furthermore, at the beginning and at
the end the cache is empty. I.e., a page may be in the cache only between two
consecutive requests to this page, and either it is in the cache for the whole
interval or not at all.

Each page of size~three is requested only twice in two consecutive blocks, and
these blocks are distinct for all pages of size three. Thus, a~service of
edge-pages is valid if and only if at each time, at most $m$~edge-pages are in
the cache. It is thus convenient to think of the cache as of $m$~\emph{slots}
for edge-pages.

Each vertex-page is requested twice. Thus, the savings on the $n$ vertex-pages
are at most~$n/(n+1)<1$. Since all edge-pages have cost one, the optimal
service must serve them optimally.  Furthermore, a vertex-page can be cached if
and only if during the phase it never happens that at the same time all slots
for edge-pages are full and a page of size three is cached.

Let $S_B$ denote the set of all edge-pages cached at the beginning of the
block~$B$ and let $s_B = |S_B|$.  Now observe that each edge-page is requested
only in a contiguous segment of blocks, once in each block. It follows that the
total savings on edge-pages are equal to $\sum_B s_B$ where the sum is over all
blocks. In particular, the maximal possible savings on the edge-pages are
$(d-1)m$, using the fact that $S_I$ is empty.

We prove that there is a~service with the total savings of at least
$(d-1)m+K/(n+1)$ if and only if there is an independent set of size $K$ in $G$.
First the easy direction.

\begin{lemma}
\label{l:simple_easy_direction}
Suppose that $G$ has an independent set $W$ of size $K$. Then there exists a
service with the total savings of $(d-1)m+K/(n+1)$.
\end{lemma}
\begin{proof}
For any $e$, denote $e=uv$ so that $u$ precedes $v$ in the ordering of phases.
If $u\in W$, we keep $\aa^e, b^e, \bb^e$ and $\beta^e$ in the cache from the
first to the last request on each page, and we do not cache $a^e$ and
$\alpha^e$ at any time. Otherwise we cache $\bb^e, a^e, \aa^e$ and $\alpha^e$,
and do not cache $b^e$ and $\beta^e$ at any time. In both cases, at each time
at most one page associated with $e$ is in the cache and the savings on those
pages is $(d-1)m$. See Fig.~\ref{fig:simple-is-easy-direction} for an
illustration.

\begin{figure}
\begin{center}
\includegraphics[scale=1]{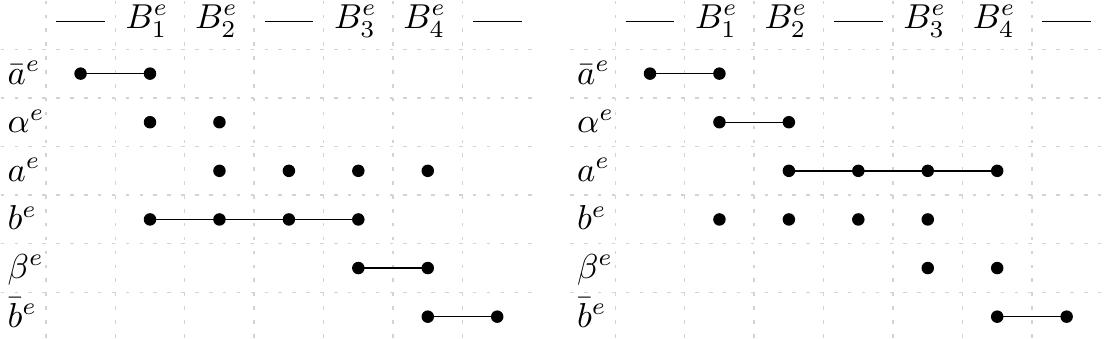}
\caption{The two ways of caching in Lemma~\ref{l:simple_easy_direction}}
\label{fig:simple-is-easy-direction}
\end{center}
\end{figure}

For any $v\in W$, we cache $p_v$ between its two requests.  To check that this
is a valid service, observe that if $v\in W$, then during the corresponding
phase no page of size three is cached. Thus the page $p_v$ always fits in the
cache together with at most $m$ pages of size two.
\qed
\end{proof}

Now we prove the converse in a sequence of claims. Fix a valid service with
savings at least $(d-1)m$. For a block~$B$, let $B'$ denote the following
block.

\begin{myclaim}
For any block $B$, with the exception of $B=I$, we have $s_B=m$.
\end{myclaim}
\begin{proof}
For each $B\neq I$ we have $s_B\leq m$. Because $s_I = 0$, the total savings on
edge-pages are $\sum_B s_B\leq (d-1)m$. We need an equality.
\qed
\end{proof}

We now prove that each edge occupies exactly one slot during the service.

\begin{myclaim}
For any block $B\neq I$ and for any $e$, $S_B$ contains exactly one page
associated with $e$.
\end{myclaim}
\begin{proof}
Let us use the notation $S_B^{\leq k} = S_B^{e_1} \cup \cdots \cup S_B^{e_k}$
and $s_B^{\leq k} = \Bigl|S_B^{\leq k}\Bigr|$. First, we shall prove for
each~$k \leq m$
\begin{equation}
\label{e:prfx}
s_B^{\leq k} = k.
\end{equation}

This is true for $B=F$, as only the $m$ edge-pages $\bb^e$ can be cached there,
and by the previous claim all of them are indeed cached. Similarly for $B=I'$
(i.e., immediately following the initial block).

If \eqref{e:prfx} is not true, then for some $k$ and $B\not\in\{I,F\}$ we have
$s_B^{\leq k}<s_{B'}^{\leq k}$. Then after processing the edge $e_k$ in the
block $B$ we have in the cache all the pages in $(S_B \setminus S_B^{\leq
k})\cup S_{B'}^{\leq k}$. Their number is $(m-s_B^{\leq k})+s_{B'}^{\leq k}>m$,
a contradiction.

The statement of the claim is an immediate consequence of~\eqref{e:prfx}.
\qed
\end{proof}

\begin{myclaim}
For any edge $e$, at least one of the pages $\alpha^e$ and $\beta^e$ is cached
between its two requests.
\end{myclaim}
\begin{proof}
Assume that none of the two pages is cached. It follows from the previous claim
that $b^e\in S_{B^e_2}$, as at this point $\alpha^e$ and $b^e$ are the only
pages associated with $e$ that can be cached. Similarly, $a^e\in S_{B^e_4}$.

It follows that there exists a block $B$ between $B^e_1$ and $B^e_4$ such that
$S_B$ contains the page $b^e$ and $S_{B'}$ contains the page $a^e$. However, in
$B$, the page $a^e$ is requested before the page~$b^e$. Thus at the point
between the two requests, the cache contains two pages associated with $e$,
plus one page associated with every other edge, the total of $m+1$ pages, a
contradiction.
\qed
\end{proof}

Now we are ready to complete this direction.

\begin{lemma}
\label{l:easy_hard_direction}
Suppose that there exists a valid service with the total savings of
$(d-1)m+K/(n+1)$. Then $G$ has an independent set $W$ of size $K$.
\end{lemma}

\begin{proof}
Let $W$ be the set of all $v$ such that $p_v$ is cached between its two
requests. The total savings imply that $|W|=K$.

Now we claim that $W$ is independent. Suppose not, let $e=uv$ be an edge with
$u,v\in W$. Then $p_u$ and $p_v$ are cached in the corresponding phases. Thus
neither $\alpha^e$ nor $\beta^e$ can be cached, since together with other $m-1$
requests of size 2 associated with the remaining edges, the cache size needed
would be $2m+2$. However, this contradicts the last claim.
\qed
\end{proof}

Lemmata~\ref{l:simple_easy_direction} and \ref{l:easy_hard_direction} together
show that we constructed a~valid polynomial-time reduction from the problem of
independent set to general caching. Therefore,
Theorem~\ref{thm:simple_hardness} is proven.


\end{document}